\documentclass[a4paper,oneside,11pt]{article}
\usepackage{a4wide, dsfont}
\usepackage{graphicx}
\usepackage[latin1]{inputenc}
\usepackage{amsmath,amssymb,amsthm}
\usepackage{url}
\usepackage{ifthen}
\usepackage{natbib}
\usepackage{color}
\usepackage{subfigure}
\usepackage{epstopdf}
\usepackage{caption}
\theoremstyle{plain}
\newtheorem{theorem}{Theorem}%
\newtheorem{cor}[theorem]{Corollary}%
\newtheorem{prop}[theorem]{Proposition}%
\theoremstyle{definition}
\newtheorem*{remark*}{Remark}%
\renewcommand\theta{\vartheta}

\newcommand{\N}{\ensuremath{\mathbb{N}}}%
\newcommand{\R}{\ensuremath{\mathbb{R}}}%
\newcommand{\cF}{\ensuremath{\mathcal{F}}}%
\DeclareMathOperator{\prob}{P}
\DeclareMathOperator{\sign}{sign}

\DeclareMathOperator*{\argmin}{arg\,min}

\renewcommand\theta{\vartheta}

\newboolean{comment}
\newcommand{\Kommentar}[1]{%
        \ifthenelse{\boolean{comment}}{#1}{}
        }
\setboolean{comment}{true}

\begin{document}
\bibliographystyle{natbib}
%
%
{\title{Expectile Asymptotics}
%
%
\author{Hajo Holzmann \\
        \small{Fachbereich Mathematik und Informatik}  \\
        \small{Philipps-Universität Marburg} \\
        \small{holzmann@mathematik.uni-marburg.de}
    \and
        Bernhard Klar \\
        \small{Institut für Stochastik}  \\
        \small{Karlsruher Institut für Technologie (KIT)} \\
        \small{bernhard.klar@kit.edu} }
}
\maketitle
\begin{abstract}
We discuss in detail the asymptotic distribution of sample expectiles. First, we show uniform consistency under the assumption of a finite mean. In case of a finite second moment, we show that for expectiles other then the mean, only the additional assumption of continuity of the distribution function at the expectile implies asymptotic normality, otherwise, the limit is non-normal. For a continuous distribution function we show the uniform central limit theorem for the expectile process. If, in contrast, the distribution is heavy-tailed, and contained in the domain of attraction of a stable law with $1 < \alpha < 2$, then we show that the expectile is also asymptotically stable distributed.  Our findings are illustrated in a simulation section.
\end{abstract}
\section{Introduction}
Expectile regression, that is, regression on a parameter that generalizes the mean and characterizes the tail behaviour of a distribution, has been introduced by  \citet{newey} as an alternative to more standard quantile regression; \citet{breckling} considered regression based on more general asymmetric M-estimators. For a recent comparison between quantile and expectile regression and references see \citet{kneib}.

Let $Y$ be a random variable with distribution function $F$ and finite mean $E|Y|<\infty$.
For a fixed $\tau \in (0,1)$, the $\tau$-expectile $\mu_{\tau} = \mu_{\tau} (F)$ of $Y$ has been introduced by \citet{newey} as the minimizer of an asymmetric quadratic loss
\begin{align}\label{eq:scoringfct}
\begin{split}
 \mu_{\tau}(F) & = \argmin_{x\in \mathbb{R}} E\, S_\tau(x,Y),\\
 S_\tau(x,y) & = \tau/2 \, \big[\, ((y-x)^{+})^{2} - (y^+)^2\big]+(1-\tau)/2\, \big[((y-x)^{-})^{2}- (y^-)^2\big].
\end{split}
\end{align}
Apparently, for $\tau = 1/2$ one obtains the mean. Alternatively to $S_\tau(x,y)$ in (\ref{eq:scoringfct}), one may use other scoring functions for the expectile; these were recently characterized by \citet[Theorem 10]{gneiting}.

Compared to quantiles, expectiles require the existence of a first moment and hence lack robustness. On the other hand, for any distribution with finite mean, the expectile is unique for each $\tau$, and the expectile curve is always strictly increasing and continuous.

More importantly, as a risk measure it has been shown recently that expectiles have the attractive property of coherence (see \citet{bellini}), while quantiles suffer from the lack of subadditivity. Indeed, expectiles were shown to be the only coherent, elicitable risk measures in \citet{ziegel14}; for a discussion and comparison between value at risk (quantiles), expectiles and expected shortfall see \citet{emmer}. Further discussion and application of expectiles as risk measures are given in \citet{delbaen} and \citet{bellini3}.

In this note we study in detail the statistical, that is, asymptotic properties of the sample expectiles. Somewhat surprisingly and in contrast to the mean, for $\tau \not=1/2$ we find that even under the assumption of a finite second moment, the sample expectile is only asymptotically normal if the distribution function $F$ is continuous at $\mu_\tau(F)$, otherwise, the limit distribution is non-normal.

First, in Section \ref{sec2} we show uniform consistency under the assumption of a finite mean. Next, in Section \ref{sec4} we show that if the distribution function $F$ is continuous at its $\tau$-expectile $\mu_\tau(F)$, there is an asymptotic linearization of the sample expectile for this $\tau$. In case of finite second moments, this implies asymptotic normality, but if $F$ is in the domain of attraction of a stable law, the sample expectile is also asymptotically stable distributed.
If $F$ has a jump at $\mu_\tau(F)$, we show in Section \ref{sec3} that also under the assumption of a finite second moment, the asymptotic distribution of the sample expectile is non-normal. Finally, for a continuous distribution function with second moments, we show the uniform central limit theorem for the expectile process.
We illustrate our findings in a simulation in Section \ref{sec:sims}, using the $t$-distribution with low degrees of freedom as a prototypical example for heavy-tailed distributions.
Based on an explicit representation of the expectile for discrete distributions, we exemplify the nonstandard asymptotic behavior of the empirical expectile by a three-point distribution. Proofs are deferred to Section \ref{sec:proofs}.

In a recent paper, \citet{kraetschmer} obtained results on the asymptotics of expectiles which are to some extend complementary to our results. Using a non-standard version of the functional delta-method allows them to treat both the case of dependent data as well as expectiles of parametric estimates of the distribution. However, they only consider the case of a finite second moment (they even assume slightly more) and a distribution which is continuous at the expectiles, and further do not investigate properties of the expectile process.
%
%
\section{Asymptotic properties of sample expectiles}
\citet{newey} state a number of useful properties of expectiles, mainly for absolutely continuous distributions $F$. Below we state an extension, and in particular point out the assumptions on $F$ which are actually required.
Introducing the identification function
\begin{equation}\label{eq:identfct}
I_\tau(x,y) = \tau (y-x) 1_{ \{y \geq x\} } - (1-\tau)\, (x-y) 1_{ \{y < x\} }
\end{equation}
of the expectile, it is well-known that
$\mu_{\tau}(F)$ can equivalently be defined as unique solution of the first-order condition
\begin{align} \label{expectile-foc}
E I_\tau(x,Y) = 0, \quad x \in \R.
\end{align}
The following identity, obtained by a partial integration, is important for us:
\begin{align}\label{eq:partialint}
I_\tau(x,F) := & E I_\tau(x,Y)
		 = \tau \int_x^\infty \big(1 - F(y) \big)\, dy - (1-\tau) \int^x_{-\infty} \,F(y) \, dy.
\end{align}

\begin{prop}\label{lem:continuity}
Let $F$ be a distribution function with finite mean. \\
(i) For each $\tau \in (0,1)$ there is a unique solution $\mu_\tau(F)$ to (\ref{eq:scoringfct}) or, equivalently, to (\ref{expectile-foc}).\\
(ii) The function $\mu_\cdot(F): (0,1) \to \R$, $\tau \mapsto \mu_\tau(F)$, is continuous, strictly increasing, and has range $ \{y\in \R:\ 0 < F(y) < 1\}$. \\
(iii) If $F$ is continuous in a neighborhood of $\mu_{\tau}(F)$ for a given $\tau \in (0,1)$, then $\mu_{\cdot}(F)$ is continuously differentiable in a neighborhood of $\tau$ with derivative
\[ \partial_\tau\, \mu_{\tau}(F) = \frac{\int_{\mu_{\tau}}^\infty \big(1 - F(y) \big)\, dy +\int^{\mu_{\tau}}_{-\infty} \,F(y) \, dy}{\tau \big(1 - F\big(\mu_{\tau}\big) \big) + (1-\tau)\, F\big(\mu_{\tau}\big)} \, . \]
\end{prop}
\subsection{Sample expectiles and uniform consistency}\label{sec2}

In this section we show strong uniform consistency of sample expectiles.
Let $Y$ have distribution function $F$, with finite first moment $E_F|Y|=E|Y|<\infty$,
and let $Y_1, Y_2, \ldots $ be i.i.d. copies of $Y$, and let $\hat F_n$ be the empirical distribution function.
The empirical $\tau$-expectile
\[
 \hat \mu_{\tau,n} = \mu_\tau\big(\hat F_n\big)
\]
can be defined as solution of the equation
\begin{align} \label{Z-estimator}
I_\tau\big(x,\hat F_n\big) = \frac1n \sum_{k=1}^n I_\tau(x, Y_k) = 0.
\end{align}
This type of estimator is often termed Z-estimator, and a large amount of theory is available to obtain asymptotic properties for this type of estimators.
Alternatively, asymptotic results can be derived using the representation as an M-estimator, that is,
\begin{align} \label{M-estimator}
 \hat \mu_{\tau,n} = \text{argmin}_{x \in \R}\,
 \hat S_n(x),\qquad \hat S_n(x) = \frac{1}{n}\, \sum_{k=1}^n S_\tau(x, Y_k) = \int_\R S_\tau(x,y)\, d \hat F_n(y).
\end{align}
Here, any other scoring function for the expectile \citep{gneiting} could be used instead, they all result in the same estimator, the expectile of the empirical distribution function.

The measurability of $\hat \mu_{\tau,n}$ follows from Theorem (1.9) in \citet{pfanzagl}, who studied M-estimators under the heading of minimum contrast estimation. More directly, measurability follows from the explicit representation of $\hat \mu_{\tau,n}$ in Subsection \ref{sec32}.

\begin{theorem} \label{prop:uniformconsistency}
Let $Y,Y_1,Y_2,\ldots$ be i.i.d.~with distribution function $F$, and assume $E_F|Y|<\infty$. For any $\tau_l, \tau_u \in (0,1)$, $\tau_l < \tau_u$, we have
\begin{align*}
\sup_{\tau_l \leq \tau \leq \tau_u}\, \big| \hat{\mu}_{\tau,n} - \mu_{\tau}(F)\big| \ \to 0 \quad a.s.
\end{align*}
\end{theorem}
\subsection{Asymptotic linearization and convergence to stable distributions} \label{sec4}
Let us consider the representation (\ref{M-estimator}) of the sample expectile as an M-estimator. Asymptotic normality or, more generally, asymptotic linearization, requires that the asymptotic contrast function has a second order Taylor expansion at the true parameter.
Since $\left|\partial_x S_\tau(x,y)\right| = \left|I_\tau(x,y)\right| \le c (|x|+|y|)$ for a suitable constant $c$,
we may differentiate the {\sl asymptotic contrast function}
\begin{align}\label{eq:asympcontrast}
 \psi_\tau(x)  = ES_\tau(x,Y) =  \int S_\tau(x,y) \, d F(y)
\end{align}
under the integral sign to obtain
\begin{align*}
 \psi_\tau'(x)  = - \, E I_\tau(x,Y) = : - I_\tau(x,F).
\end{align*}
We see from (\ref{eq:partialint}) that $ \psi_\tau'(x)$ has
\begin{align}\label{eq:derivatives}
\begin{split}
\text{right derivative}\quad \psi^{''+}_\tau(x) = & \, \tau \big(1 - F(x) \big) + (1-\tau)\, F(x)\\
\text{left derivative}\quad \psi^{''-}_\tau(x) = &  \, \tau \big(1 - F(x-) \big) + (1-\tau)\, F(x-)
\end{split}
\end{align}
at $x$, where $F(x-) = \prob(Y < x)$ is the left limit of $F$ at $x$. For $\tau= 1/2$ (i.e.~the mean), these are always equal, but generally only coincide at $\mu_\tau(F)$ if $F$ has no point mass in its $\tau$-expectile. From Theorems 1 and 10 in \citet{arcones} we deduce the following linearization.

\begin{theorem}[{\sl Asymptotic linearization}]\label{the:asymplinearization}
Let $Y,Y_1,Y_2,\ldots$ be i.i.d.~with distribution function $F$. Assume that $E_F|Y|<\infty$ and that $F$ is continuous at $\mu_\tau = \mu_\tau(F)$ for a given $\tau \in (0,1)$. Let $\{a_n\}$ be a sequence of positive numbers which converges to infinity with
$\sup_{n\geq 1} n^{-1} a_n^2 < \infty$, such that
\begin{align}\label{eq:boundedness}
\frac{a_n}{n}\,\sum_{k=1}^n \, I_\tau\left(\mu_{\tau}, Y_k \right) =O_{\prob}(1).
\end{align}
Then
\begin{equation}\label{eq:asymplinear}
a_n \left(\hat \mu_{\tau,n} -  \mu_{\tau} \right) = \frac{a_n}{n}\, \big(\tau \big(1 - F(\mu_{\tau}) \big) + (1-\tau)\, F(\mu_{\tau})\big)^{-1} \sum_{k=1}^n \, I_\tau\left(\mu_{\tau}, Y_k \right) + o_{\prob}(1).
\end{equation}
\end{theorem}
\vspace{2mm}
{\sl Asymptotic normality}\\*[0.1cm]
In case of finite second moments, (\ref{eq:boundedness}) is satisfied with $a_n = \sqrt{n}$ by the central limit theorem, and we obtain asymptotic normality for a finite number of expectiles. In the following, we write $Y_n\stackrel{\mathcal L}{\to} F$ as a short-hand notation for $Y_n\stackrel{\mathcal L}{\to} Y \sim F$, where $F$ denotes the distribution function of $Y$.
\begin{cor}\label{cor:asympnorm}
Suppose that $E Y^2 < \infty$. Let $\tau_i \in (0,1)$, $i=1, \ldots, m$ be such that $F$ does not have a point mass at any of the $\mu_{\tau_i}$, $i=1, \ldots, m$. Then
\[ \sqrt{n} \big(\hat \mu_{\tau_1,n} -  \mu_{\tau_1}, \ldots, \hat \mu_{\tau_m,n} -  \mu_{\tau_m} \big)' \stackrel{\mathcal L}{\to} N\big( \mathbf{0}, \Sigma\big),\]
where
\begin{equation}\label{eq:covariance}
\Sigma_{i,j} = \frac{E\left[I_{\tau_i}\left(\mu_{\tau_i}, Y \right)\, I_{\tau_j}\left(\mu_{\tau_j}, Y \right) \right]}{\big(\tau_i \big(1 - F(\mu_{\tau_i}) \big) + (1-\tau_i)\, F(\mu_{\tau_i})\big)\, \big(\tau_j \big(1 - F(\mu_{\tau_j}) \big) + (1-\tau_j)\, F(\mu_{\tau_j})\big)}
\end{equation}
for $i,j=1, \ldots, m$.
%
\end{cor}
\vspace{2mm}
{\sl Convergence to stable distributions}\\*[0.1cm]
A random variable $X$ has an $\alpha$-stable distribution if its characteristic function is given by
\begin{align*}
E\left[ e^{iuX} \right] &=
  \begin{cases}
   \exp\left( - |u|^{\alpha} \left[ 1-i\beta \tan\left(\frac{\pi\alpha}{2}\right) \sign(u) \right]\right), & \alpha \neq 1, \\
   \exp\left( - |u| \left[ 1+i\beta \frac{2}{\pi} \sign(u) \log |u| \right] \right),      & \alpha = 1,
  \end{cases}
\end{align*}
where $0<\alpha\leq2, \ \beta\in[-1,1]$. Assume that $Y$ belongs to the {\sl domain of attraction} of an $\alpha$-stable distribution ($Y\in DA(\alpha)$) with $0<\alpha<2$ (see, e.g., \citet[Def. 2.2.7]{embrechts}). This is the case if and only if $Y$ has tail probabilities that satisfy
\begin{align}\label{eq:tailprobstable}
 P(Y > y) = \frac{c^+ + o(1)}{y^\alpha} L(y) \quad \mbox{and} \quad P(Y < -y) = \frac{c^- + o(1)}{y^\alpha} L(y), \quad  y\to\infty,
\end{align}
where $L$ is slowly varying and $c^+,c^- \geq 0$ with $c^+ + c^- > 0$ \citep[Th. 2.2.8]{embrechts}. In the following, we assume $1<\alpha<2$ to ensure that $E|Y|<\infty$.

\begin{cor} \label{theorem-stable}
Let $Y,Y_1,Y_2,\ldots$ be i.i.d. r.v. with distribution function $F \in DA(\alpha)$, where $1<\alpha<2$.
Assume further that $F$ has no point mass in $\mu_\tau$. Then,
\begin{align*}
\frac{n^{1-1/\alpha}}{L_1(n)}\,  \left( \hat \mu_{\tau,n}-\mu_{\tau}(F)\right) &\xrightarrow{\cal{L}} \frac{\tilde{Z}}{\tau \big(1 - F(\mu_{\tau}) \big) + (1-\tau)\, F(\mu_{\tau})}.
\end{align*}
Here, $\tilde{Z}$ follows an $\alpha$-stable distribution, and $L_1$ is an appropriate slowly varying function.
\end{cor}
\begin{proof}
Since $F \in DA(\alpha)$, from (\ref{eq:tailprobstable}) we obtain that
\begin{align*}
 P(I(\mu_{\tau},Y)>y)  & = \tau^{\alpha} \, \frac{c^+ + o(1)}{y^\alpha} L(y)    \quad \mbox{and} \\
 P(I(\mu_{\tau},Y)<-y) & = (1-\tau)^{\alpha} \, \frac{c^- + o(1)}{y^\alpha} L(y)   \quad  \mbox{as } y\to\infty.
\end{align*}
Consequently, $I(\mu_{\tau},Y) \in DA(\alpha)$, and the general CLT \citep[Th. 2.2.15]{embrechts} yields
 \[ \left( n^{1/\alpha} L_1(n) \right)^{-1} \left( \sum_{k=1}^n I(\mu_{\tau},Y_k) - nEI(\mu_{\tau},Y) \right) \ \xrightarrow{\mathcal L} \ \tilde{Z}
 \quad \mbox{as } n\to\infty,
\]
where $\tilde{Z}$ follows an $\alpha$-stable distribution and $L_1$ is an appropriate slowly varying function. This implies that (\ref{eq:boundedness}) is satisfied, and an application of Theorem \ref{the:asymplinearization} together with the general CLT yields the statement of the corollary.
\end{proof}
Instead of using the assumptions $Y\in DA(\alpha)$, suppose more specifically that $Y$ belongs to the {\sl domain of normal attraction} of some $\alpha$-stable distribution with $1<\alpha<2$, i.e. $Y$ has tail probabilities that satisfy
\begin{equation}\label{eq:normalattract}
y^{\alpha}P(Y > y)\to c^+ \quad \text{and}\quad  y^{\alpha}P(Y < -y)\to c^-,\qquad y\to\infty,
\end{equation}
with $c^+ + c^- > 0$ and $1<\alpha<2$.
\begin{cor} \label{theorem-stable2}
Let $Y,Y_1,Y_2,\ldots$ be i.i.d. r.v. with distribution function $F$ that belongs to the normal domain of attraction of an $\alpha$-stable distribution, where $1<\alpha<2$, that is, satisfies (\ref{eq:normalattract}).
Assume further that $F$ has no point mass in $\mu_\tau$. Then
\begin{align*}
 n^{1-1/\alpha} \, \tilde{c} \, \left( \hat \mu_{\tau,n}-\mu_{\tau}\right) &\xrightarrow{\cal{L}}
 \frac{ S(\alpha,\tilde{\beta}) }{\tau \big(1 - F(\mu_{\tau}) \big) + (1-\tau)\, F(\mu_{\tau})},
\end{align*}
where
\begin{align*}
  \tilde{c} = \left( \frac{2\Gamma(\alpha)\sin(\pi\alpha/2)}{\pi(\tau^{\alpha} c^+ + (1-\tau)^{\alpha} c^-)} \right)^{1/\alpha},
  \quad \tilde{\beta} = \frac{\tau^{\alpha} c^+ - (1-\tau)^{\alpha} c^-}{\tau^{\alpha} c^+ + (1-\tau)^{\alpha} c^-)} \, .
\end{align*}
\end{cor}
\begin{proof}
Since
\begin{align*}
 y^{\alpha}P(I(\mu_{\tau},Y)>y)  & \to \tau^{\alpha} \, c^+     \quad \mbox{and} \quad
 y^{\alpha}P(I(\mu_{\tau},Y)<-y)  \to (1-\tau)^{\alpha} \, c^-   \quad  \mbox{as } y\to\infty,
\end{align*}
this follows from the general CLT for distributions in the normal domain of attraction of a corresponding stable law \citep[p.~22]{nolan}.
\end{proof}
\subsection{Further asymptotics under finite second moments} \label{sec3}
Suppose that $Y \sim F$ with $E\, Y^2 < \infty$ and $Var\, Y >0$. In contrast to the mean, asymptotic normality of general expectiles as in Corollary \ref{cor:asympnorm} actually requires the additional assumption that $Y$ has no point mass at $\mu_\tau(F)$, otherwise, the limit distribution is non-normal, as the following result shows.
\begin{theorem} \label{them:asymdirtsecondmom}
Let $Y,Y_1,Y_2,\ldots$ be i.i.d.~with distribution function $F$ with $E\, Y^2 < \infty$. Let $\tau \in (0,1)$ and denote $\mu_\tau = \mu_\tau(F)$. Then
\begin{align} \label{limit-distr}
  \sqrt{n} \left(\hat \mu_{\tau,n} -  \mu_{\tau} \right) \stackrel{\mathcal L}{\to} \sigma_1 \, W \, 1_{W>0} + \sigma_2 \, W \, 1_{W<0},
\end{align}
where $W \sim N\big(0, E [ I_\tau(\mu_{\tau}, Y )^2 ] \big)$,
%
%
\begin{align} \label{asymptotic-var}
\sigma_1 \ = \  \frac{1}{ \tau \big(1 - F\big(\mu_{\tau}\big) \big) + (1-\tau)\, F\big(\mu_{\tau}\big) }, \quad
\sigma_2 \ = \  \frac{1}{ \tau \big(1 - F\big(\mu_{\tau}-\big) \big) + (1-\tau)\, F\big(\mu_{\tau}-\big) },
\end{align}
and $F(x-) = P(Y<x)$ denotes the left limit of $F$ at $x$.
\end{theorem}
%
%
We prove Theorem \ref{them:asymdirtsecondmom} by using empirical process methods and the argmax continuity theorem as presented in \citet{vaart}. Alternatively one could exploit the convexity of the contrast and modify the assumptions and the proof in \citet[Theorem 2.1]{hjort} to give an alternative argument.

In case of a continuous distribution function, we also have convergence of the expectile process.
\begin{theorem}\label{satz:cltexpec}
Let $Y,Y_1,Y_2,\ldots$ be i.i.d.~with distribution function $F$ with $E\, Y^2 < \infty$. Let $ 0 < \tau_l < \tau_u < 1$ and suppose that $F$ is continuous in a neighborhood of $\big[\mu_{\tau_l}, \mu_{\tau_u} \big]$. Then the sequence of processes
\begin{equation}\label{eq:expectileprocess}
 \tau \mapsto \big(\sqrt{n} \, (\hat \mu_{\tau, n} - \mu_\tau )\big)_{n \geq 1},\qquad \tau \in  \big[\tau_l, \tau_u \big],
\end{equation}
converges weakly in $C\big[\tau_l, \tau_u \big]$ to a Gaussian process with continuous sample paths and covariance function given in (\ref{eq:covariance}).
\end{theorem}
\citet{tran} also show convergence of the expectile process. They argue via convergence of an associated quantile process, and therefore require that $F$ has a density, further, they do not specify the covariance function of the limit process.

Theorem \ref{them:asymdirtsecondmom} shows that process convergence, at least in $C\big[\tau_l, \tau_u \big]$ or even in $l^\infty\big[\tau_l, \tau_u \big]$, cannot be expected if $F$ has a discontinuity in $[\tau_l, \tau_u ]$, since in this case the limit process would be discontinuous as well.
%
%
\section{Some Simulations}\label{sec:sims}

\subsection{Illustration of convergence to a stable distribution}

As an example for a distribution with finite expectation but infinite variance we consider Student's $t$-distribution $t_\alpha, 1<\alpha<2$, with symmetric density
\begin{align*}
 f_\alpha(x) &= \frac{\Gamma((\alpha+1)/2)}{\Gamma(\alpha/2) \sqrt{\alpha\pi}} \, \left( 1+\frac{x^2}{\alpha} \right)^{-\frac{\alpha+1}{2}}, \quad x\in \R.
\end{align*}
For $Y_\alpha \sim t_\alpha$,
\begin{align*}
 \lim_{y\to\infty} y^\alpha P(Y_\alpha > y) &=  \lim_{y\to\infty} y^\alpha P(Y_\alpha < -y) = \frac{\Gamma((\alpha+1)/2)}{\Gamma(\alpha/2)} \, \frac{ \alpha^{\alpha/2-1}}{\sqrt{\pi}}.
\end{align*}
Accordingly,  $t_\alpha$ belongs to the domain of normal attraction of some $\alpha$-stable distribution.
To compute the theoretical $\tau$-expectile, which is the unique solution of
\begin{align} \label{exp1}
 \mu_{\tau}-EY &= \frac{2\tau-1}{1-\tau} \ E\left[ (Y-\mu_{\tau})^+\right],
\end{align}
one can use the identity
\begin{align*}
 E\left[ (Y_\alpha - \mu_{\tau})^+\right] = \frac{\sqrt{\alpha} \, \Gamma((\alpha+1)/2)}{\sqrt{\pi} (\alpha-1) \Gamma(\alpha/2)}
            \, \left( 1+\frac{\mu_{\tau}^2}{\alpha} \right)^{\frac{1-\alpha}{2}} - \mu_{\tau} \left(1-F_\alpha(\mu_{\tau})\right),
\end{align*}
where $F_\alpha(\cdot)$ denotes the distribution function of $t_\alpha$.
The limiting behavior of the empirical $\tau$-expectile then follows directly from  Corollary \ref{theorem-stable}.
Figure \ref{plots-stable-convergence} shows the distribution function of $n^{1-1/\alpha} \, \tilde{c} \, \left( \hat \mu_{\tau,n}-\mu_{\tau}\right)$
(more precisely the empirical distribution function based on 10000 replications) for sample sizes of 20, 200 and 2000 for several values of $\tau$ and $\alpha$.
It can be observed that the quality of the approximation by the corresponding limiting stable law depends on both $\tau$ and $\alpha$:
the approximation improves for decreasing $\alpha$ (see Figure \ref{plots-stable-convergence} (a)-(c)) and for $\tau$ approaching the value 0.5 (see Figure \ref{plots-stable-convergence} (d)-(f)).

\begin{figure}
\subfigure[$\tau=0.8, \ \alpha=1.2$]{\includegraphics[width=0.33\textwidth]{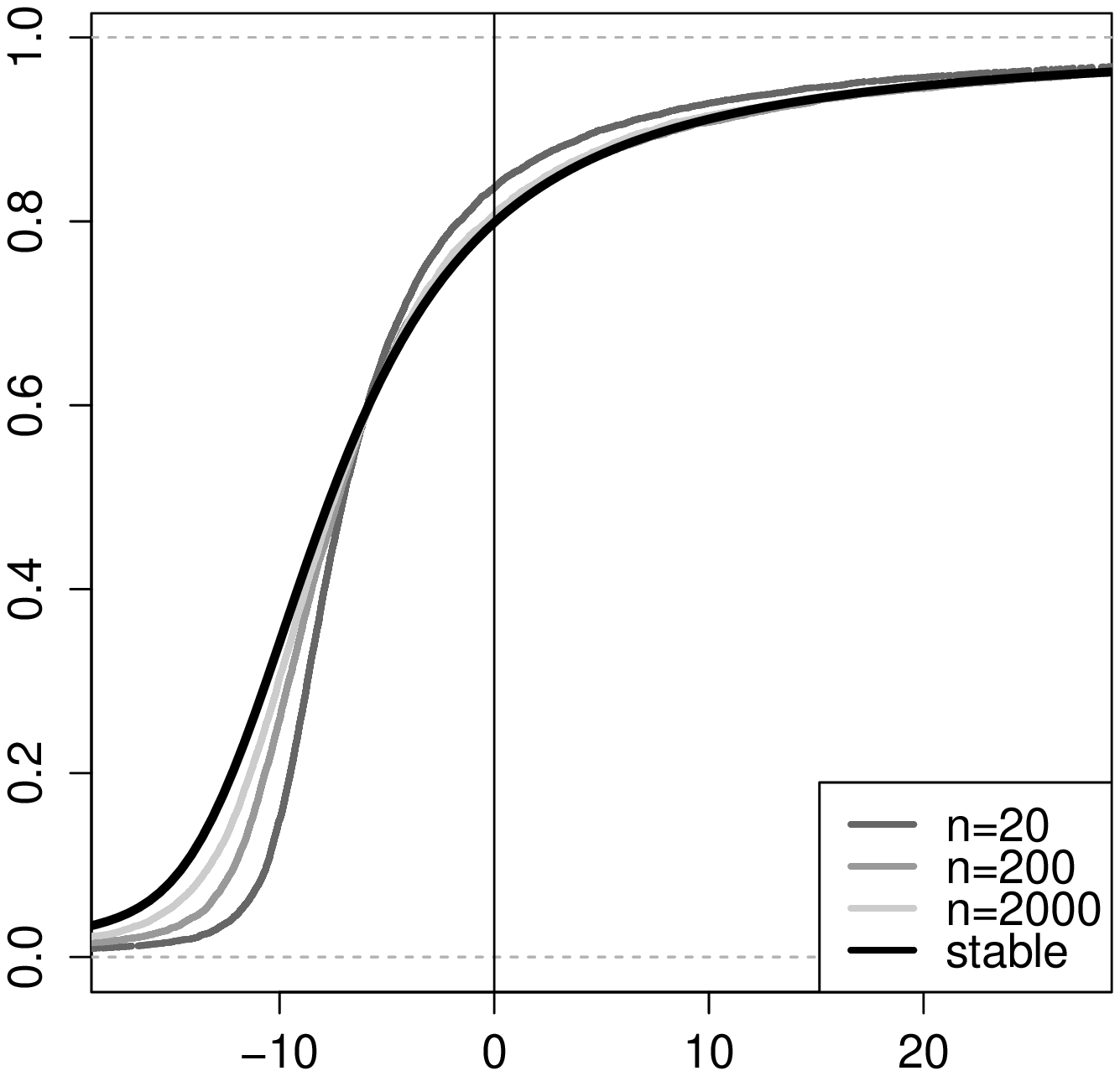}}\hfill
\subfigure[$\tau=0.8, \ \alpha=1.5$]{\includegraphics[width=0.33\textwidth]{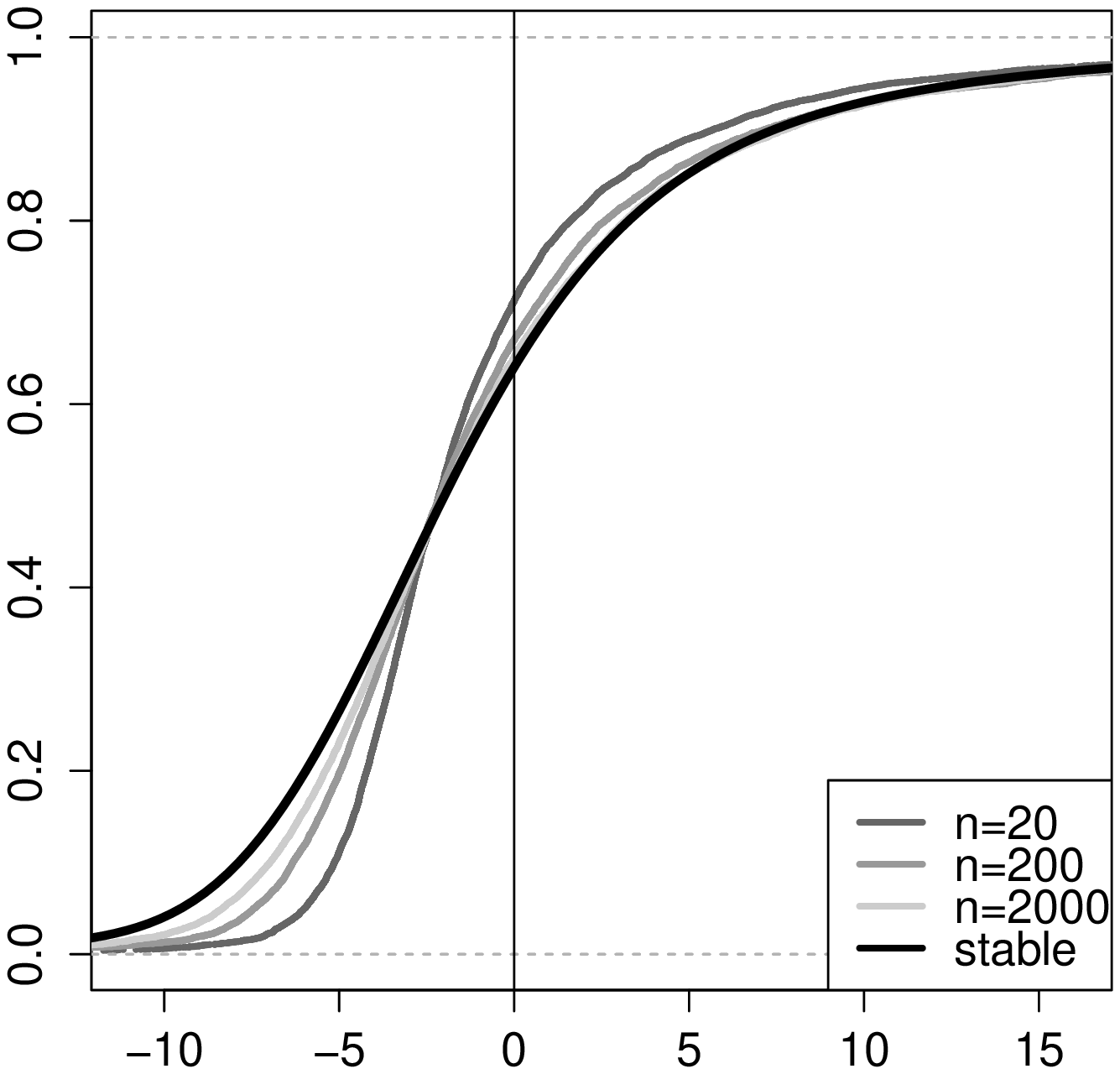}}\hfill
\subfigure[$\tau=0.8, \ \alpha=1.8$]{\includegraphics[width=0.33\textwidth]{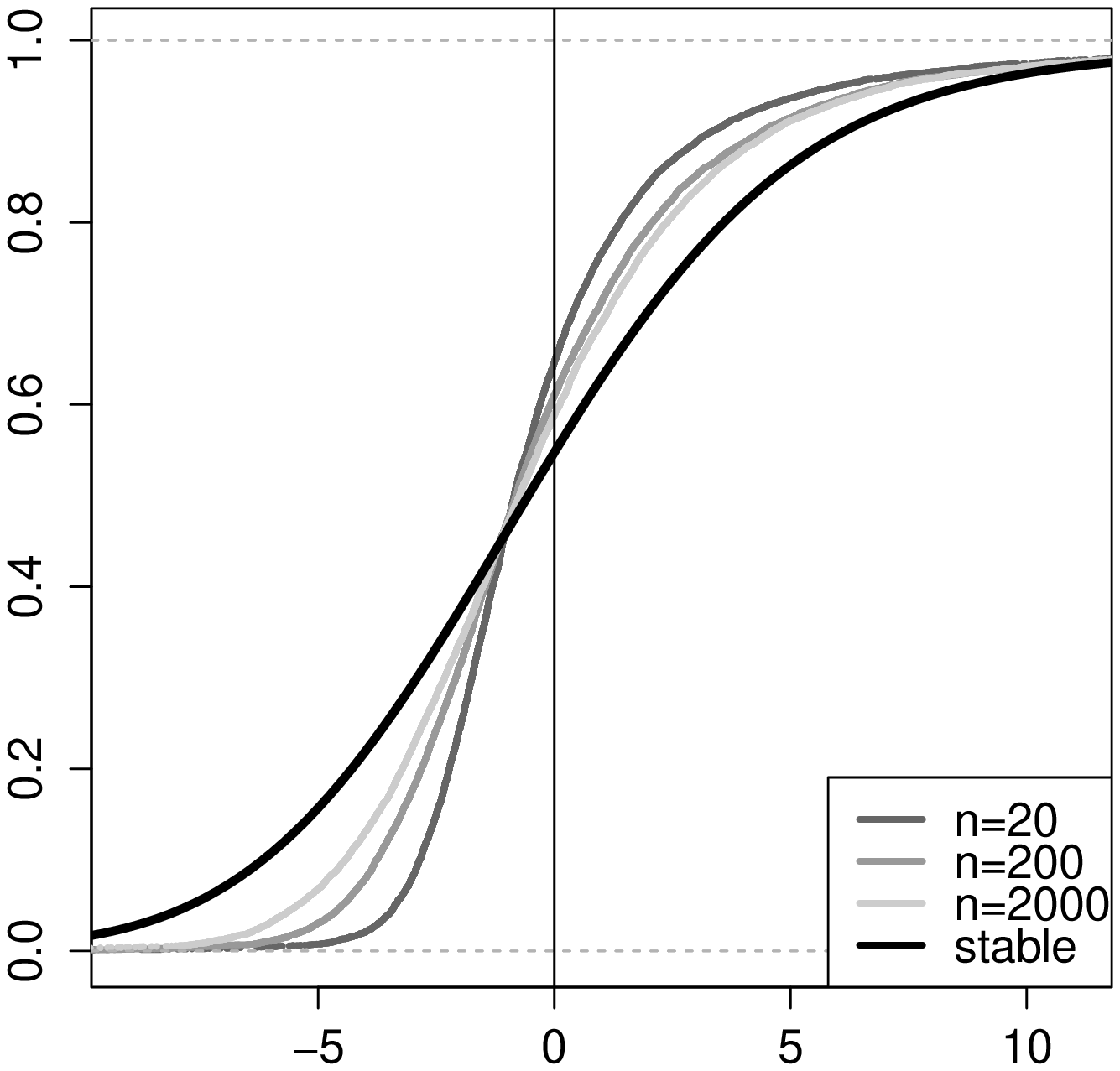}}\hfill
\subfigure[$\tau=0.6, \ \alpha=1.5$]{\includegraphics[width=0.33\textwidth]{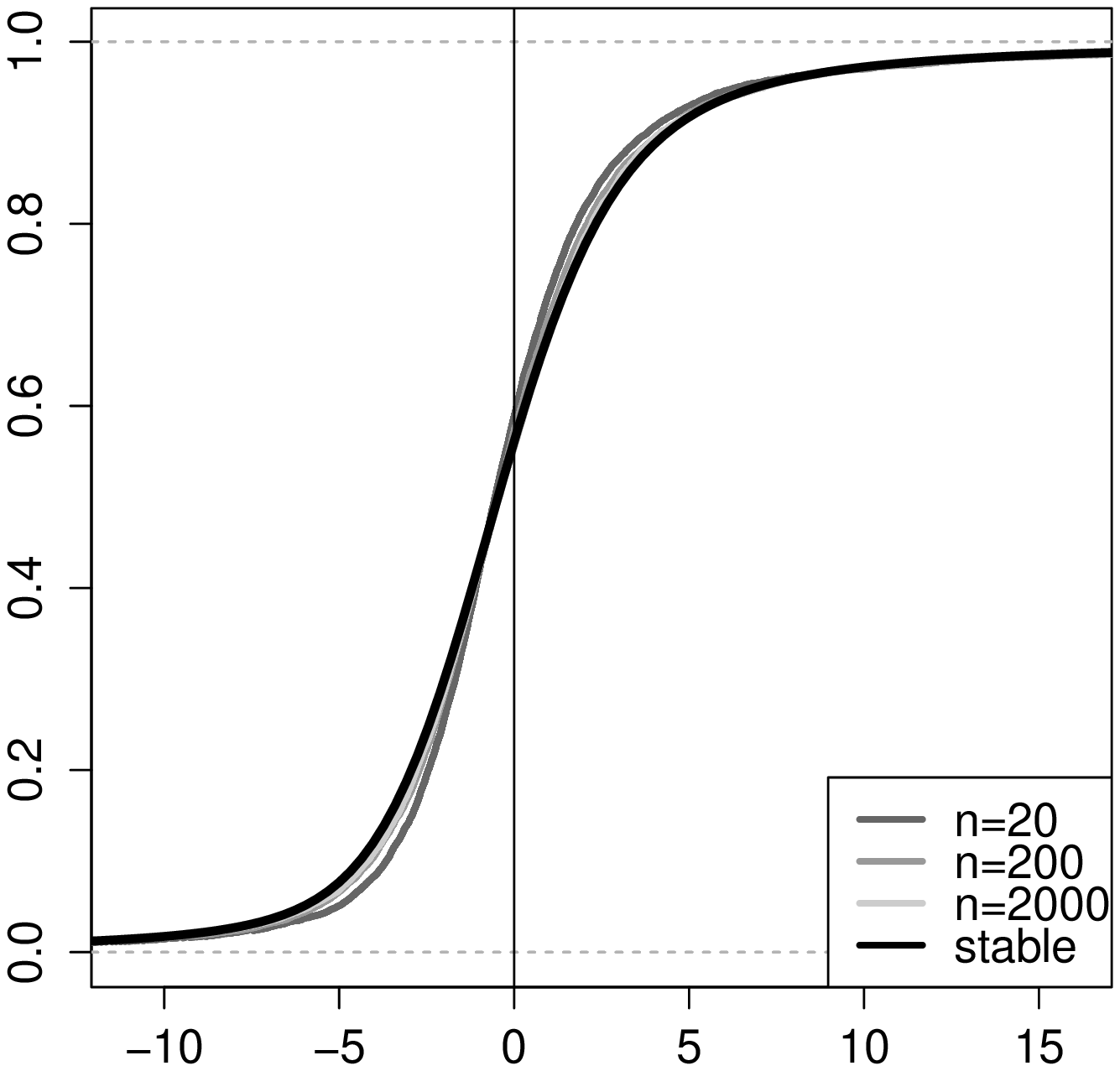}}\hfill
\subfigure[$\tau=0.9, \ \alpha=1.5$]{\includegraphics[width=0.33\textwidth]{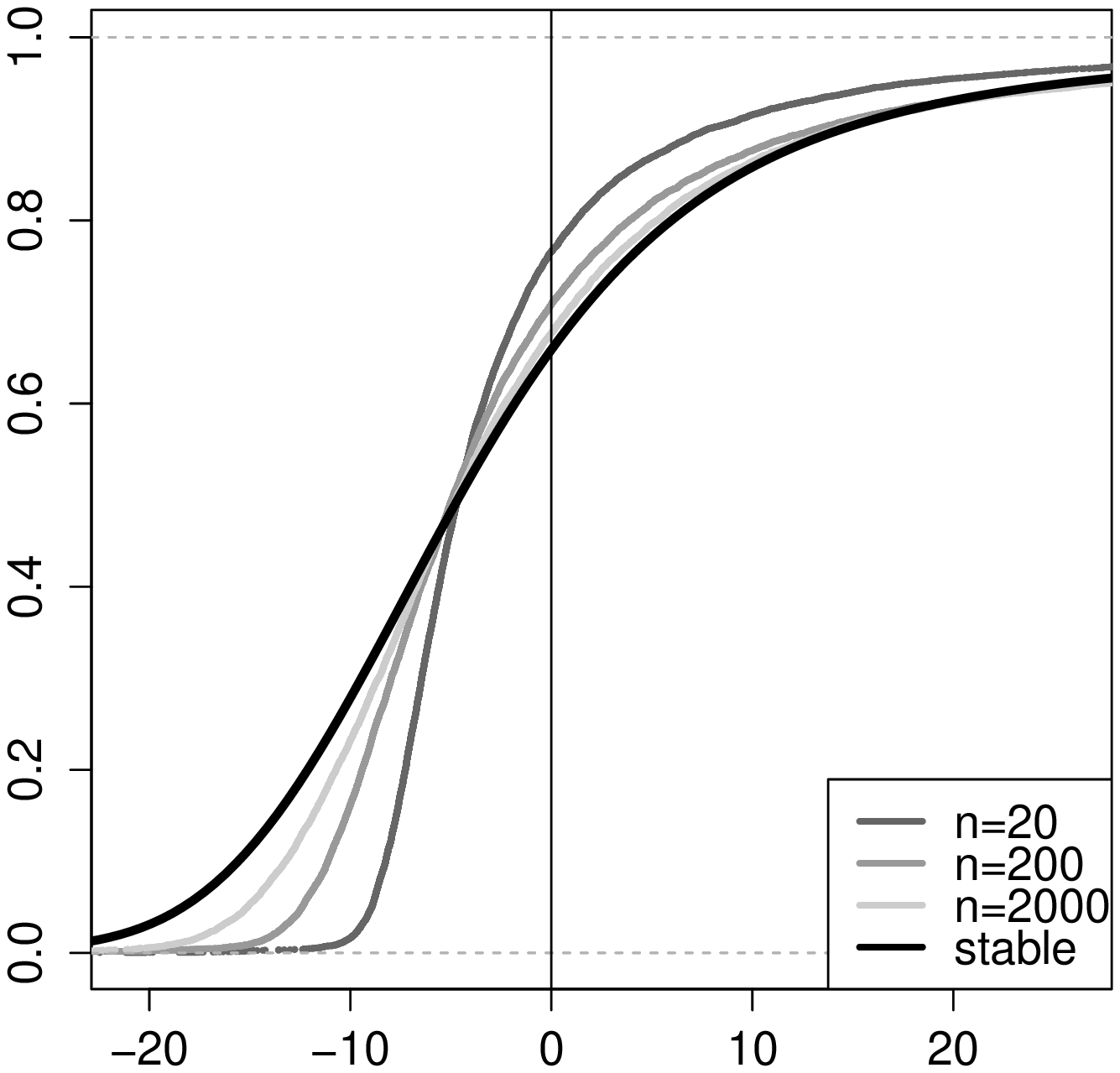}}\hfill
\subfigure[$\tau=0.95,\ \alpha=1.5$]{\includegraphics[width=0.33\textwidth]{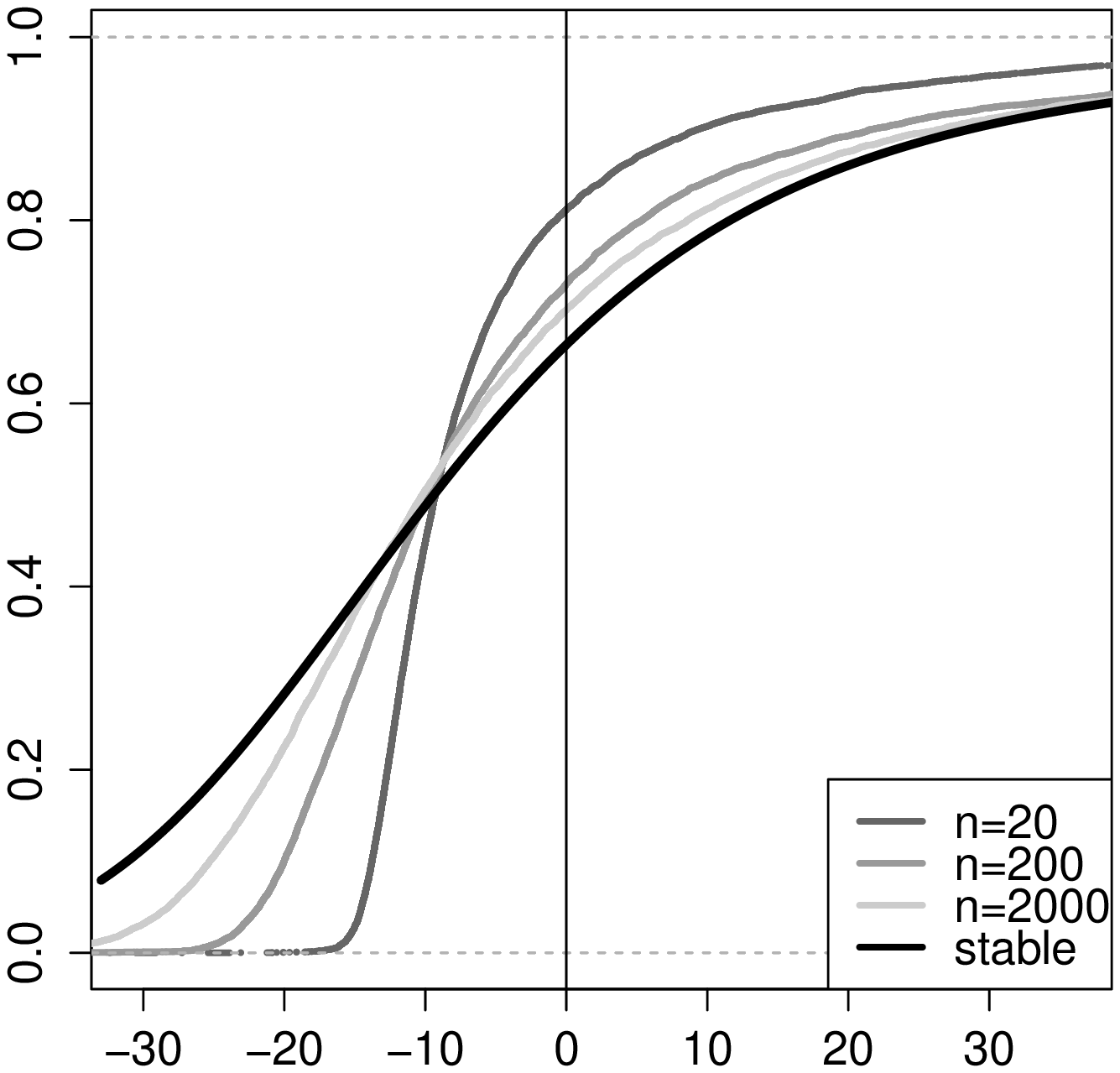}}\hfill
\caption{Convergence of the cumulative distribution function (cdf) of the empirical expectile to the corresponding limiting stable cdf. \\
Upper row: Data follow $t_{\alpha}$-distribution with different $\alpha$, $\tau=0.8$ fixed. \\
Lower row: Data follow $t_{\alpha}$-distribution with $\alpha=1.5$, different values of $\tau$.}
\label{plots-stable-convergence}
\end{figure}

\subsection{Illustration of nonstandard asymptotics under finite second moments} \label{sec32}

To illustrate the convergence to a non-normal distribution stated in Theorem \ref{them:asymdirtsecondmom}, we first give an explicit formula for the
empirical expectile which is interesting in itself. From (\ref{exp1}), it follows directly  that the $\tau$-expectile satisfies the equivalent conditions
\begin{align}
 \tau &= \frac{E\left[ (Y-\mu_{\tau})^-\right]}{E\left[ |Y-\mu_{\tau}| \right]}, \label{exp2} \\
 \mu_{\tau} &= \frac{ (1-\tau) E\left[ Y 1_{\{Y\leq \mu_{\tau}\}} \right] + \tau E\left[ Y 1_{\{Y>\mu_{\tau}\}} \right] }
                    {(1-\tau) P\left(Y \leq\mu_{\tau}\right) + \tau P\left(Y>\mu_{\tau}\right)} \ . \label{exp3}
\end{align}
The subsequent representation follows \citet{bellini2}, but formulated for the empirical distribution, and allowing for ties.
Let $Y_{(1)}\leq\ldots\leq Y_{(n)}$ denote the order statistics of $Y_1,\ldots,Y_n$. From (\ref{exp3}), the empirical expectile satisfies
\begin{align*} 
 \hat{\mu}_{\tau,n} &=
 \frac{ (1-\tau) \sum_k Y_{(k)} 1_{ \{Y_{(k)}\leq \hat\mu_{\tau,n}\} }  + \tau \sum_k Y_{(k)} 1_{ \{Y_{(k)}>\hat\mu_{\tau,n}\}} }
      { (1-\tau) \sum_k 1_{ \{Y_{(k)}\leq \hat\mu_{\tau,n}\} } + \tau \sum_k 1_{ \{Y_{(k)} > \hat\mu_{\tau,n}\}} } \ .
\end{align*}
Hence, for $\hat{\mu}_{\tau,n}\in [Y_{(i)},Y_{(i+1)})$, where $Y_{(i)}<Y_{(i+1)}$, one has
\begin{align} \label{emp-exp}
 \hat{\mu}_{\tau,n} &=
 \frac{ (1-\tau) \sum_{k=1}^i Y_{(k)} + \tau \sum_{k=i+1}^n Y_{(k)} }{ (1-\tau) i + \tau (n-i) } \ .
\end{align}
Defining
\begin{align} \label{tau*}
 \tau_i^* &:= \frac{ i Y_{(i)} - \sum_{k=1}^i Y_{(k)} }{ \sum_{k=1}^n |Y_{(k)}-Y_{(i)}| }, \quad i=1,\ldots,n,
\end{align}
we have $\hat{\mu}_{\tau,n}=Y_{(i)}$ iff $\tau=\tau_i^*$ for $i=1,\ldots,n$
(and then, (\ref{tau*}) is the empirical counterpart of (\ref{exp2})).
Note that $\tau_0^*=0, \tau_n^*=1$, and since $\hat{\mu}_{\tau,n}$ is nondecreasing in $\tau$,
we obtain that  $\tau_i^* \leq \tau_{i+1}^*, i=1,\ldots,n-1$. As a consequence,
\[
 \hat{\mu}_{\tau,n}\in [Y_{(i)},Y_{(i+1)}) \; \Leftrightarrow \; \tau \in [\tau_i^*,\tau_{i+1}^*), \quad i=1,\ldots,n-1.
\]

\begin{remark*}
\begin{enumerate}
\item
Formulas (\ref{tau*}) and (\ref{emp-exp}) are especially well-suited for plotting purposes without
the need of any numerical root-finding.
\item
From (\ref{emp-exp}), $\hat{\mu}_{\tau,n}$ is piecewise differentiable in $\tau$ with
\begin{align*}
 \frac{d \hat{\mu}_{\tau,n}}{d\tau} &=
 \frac{ i \sum_{k=i+1}^n Y_{(k)} - (n-i) \sum_{k=1}^i Y_{(k)} }{ ((1-\tau) i + \tau (n-i))^2 }
 \quad \mbox{for } \tau \in (\tau_i^*,\tau_{i+1}^*).
\end{align*}
\end{enumerate}
\end{remark*}

\vspace{3mm}
If $Y$ has a discrete distribution on $0,1,2,\ldots$ (say), an analogous reasoning leads to the following explicit formula for the theoretical expectiles $\mu_{\tau}$. Define
\begin{align} \label{discrete1}
 \tau_i^* &:= \frac{ \sum_{k=0}^{i-1} (i-k)P(Y=k) }{ \sum_{k\geq 0} |i-k| P(Y=k)}, \quad i=0,1,2,\ldots.
\end{align}
For $\tau \in [\tau_i^*,\tau_{i+1}^*)$,  and accordingly $\mu_{\tau}\in [i,i+1)$, one has
\begin{align} \label{discrete2}
 \mu_{\tau} &= \frac{ (1-\tau) \sum_{k\leq i} kP(Y=k) + \tau \sum_{k > i} kP(Y=k) }
                    { (1-\tau) P(Y \leq i) + \tau P(Y>i) } \ .
\end{align}
Now, assume that $Y$ follows a three-point distribution with $P(Y=i)=p_i, \, i=0,1,2$, with $p_0,p_1,p_2>0, p_0+p_1+p_2=1$.
Then, from (\ref{discrete1}) and (\ref{discrete2}), we get $\tau_0^*=0, \tau_1^*=p_0/(p_0+p_2), \tau_2^*=1$ and
\begin{align*}
\mu_\tau &= \begin{cases}
             \frac{\tau(p_1+2p_2)}{(1-\tau)p_0+\tau(p_1+p_2)}, & 0<\tau < \tau_1^*, \\
             \frac{(1-\tau)p_1+2\tau p_2}{(1-\tau)(p_0+p_1)+\tau p_2},     & \tau_1^* \leq \tau <1.
        \end{cases}
\end{align*}
Next, we make the choice $p_0=4/10, p_1=5/10, p_2=1/10$. Then, $\mu_{0.8}=1$, i.e. the distribution of $Y$ has a point mass in $\mu_{\tau}$ for $\tau=0.8$, but not for other values of $\tau$. Figure \ref{plots-non-normal-convergence} (a) shows the density of $ \sqrt{n} \left(\hat \mu_{\tau,n} -  \mu_{\tau} \right)$
(estimated by a nonparametric density estimator based on 20000 replications) for sample size 500 and $\tau=0.7$ (hence, $\mu_{\tau}=49/54$) together with the limiting normal distribution given in Corollary \ref{cor:asympnorm}. Figure \ref{plots-non-normal-convergence} (b) shows the corresponding plot for $\tau=0.8$ together with the limiting non-normal distribution given in Corollary \ref{them:asymdirtsecondmom}.

\begin{figure}
\subfigure[$\tau=0.7, \ \mu_{\tau}=0.907$]{\includegraphics[width=0.5\textwidth]{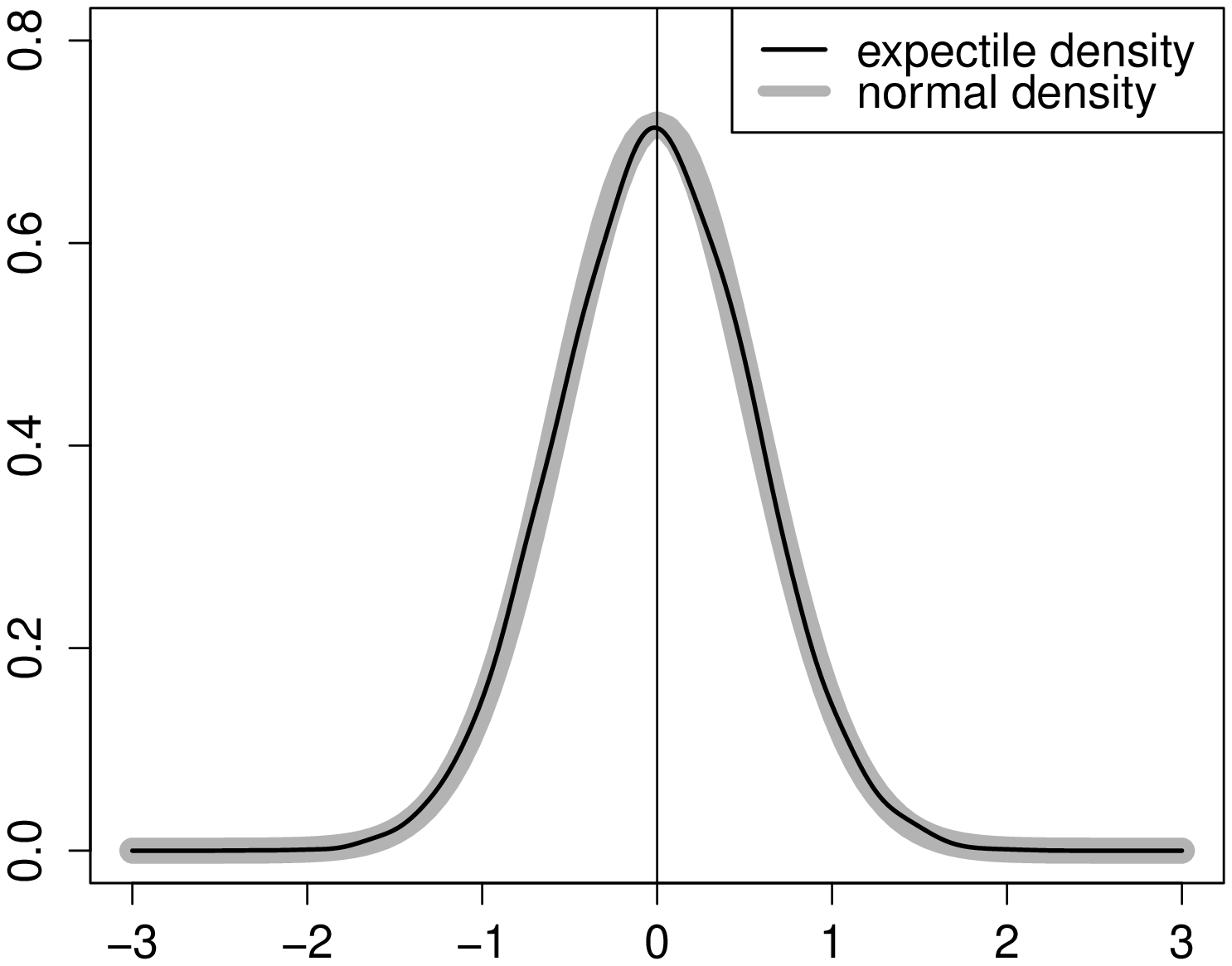}}\hfill
\subfigure[$\tau=0.8, \  \ \mu_{\tau}=1$]{\includegraphics[width=0.5\textwidth]{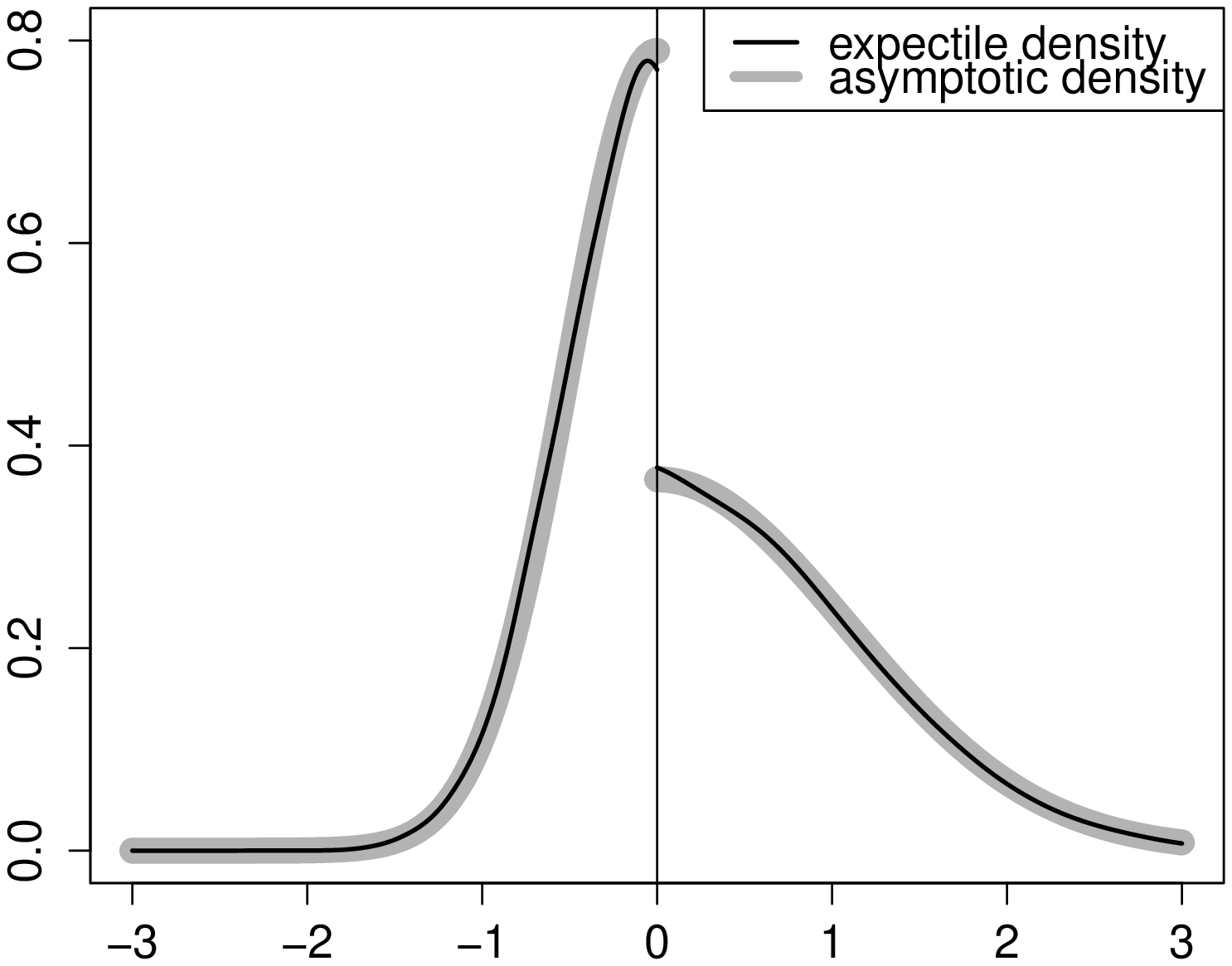}}\hfill
\caption{Density function of the standardized empirical expectile for $n=500$ and of the corresponding limiting distribution.
Data follow a three point distribution in 0,1,2.  \\
(a) $\tau=0.7$, normal limiting distribution. (b) $\tau=0.8$, non-normal limiting distribution. }
\label{plots-non-normal-convergence}
\end{figure}
%
%
\section{Proofs}\label{sec:proofs}
\begin{proof}[{\sl Proof of Proposition \ref{lem:continuity}}]
Parts (i) and (ii) are from \citet{newey} except for the general continuity of $\mu_\tau(F)$ in $\tau$. From (\ref{eq:partialint}) we see that
$I_\tau(x,F)$ is a continuous function of $(\tau,x)$. To show continuity of the expectile, first let $\tau_n \downarrow \tau$, and let $\tilde \mu_\tau = \lim_n \mu_{\tau_n} (F)$ for which by monotonicity $\mu_{\tau} (F) \leq \tilde \mu_\tau$. By continuity of  $I_\tau(x,F)$ we have
\[ 0 = \lim_n I_{\tau_n}\big(\mu_{\tau_n} (F),F\big) = I_{\tau}\big(\tilde \mu_{\tau} ,F\big),\]
but since $\mu_{\tau} (F) $ is the unique zero, it follows that $\mu_{\tau} (F) = \tilde \mu_\tau$, that is, right-continuity. The argument for left-continuity is the same. \\
(iii) \quad From (\ref{eq:partialint}) we see that if $F$ is continuous in a neighborhood of $x$, then $I_\tau(\cdot,F)$ is continuously differentiable at $x$ with derivative
$-\tau \big(1 - F(x) \big) - (1-\tau)\, F(x)$. The conclusion follows from the implicit function theorem.
\end{proof}
\begin{proof}[{\sl Proof of Theorem \ref{prop:uniformconsistency}}]
We start with strong consistency of individual expectiles, that is,
\begin{align}\label{eq:consistindividual}
  \hat{\mu}_{\tau,n} \xrightarrow{} \mu_{\tau}(F)\quad a.s.
\end{align}
We may use the representation (\ref{Z-estimator}) of the empirical expectile as a Z-estimator and strengthen \citet[Lemma 5.10]{vaart} to almost sure convergence. Since $x \mapsto I_\tau(x,F)$ is strictly decreasing, we have for every $\varepsilon>0$ that
\[ I_\tau(\mu_{\tau}-\varepsilon, F) > 0 > I_\tau(\mu_{\tau}+\varepsilon,F).\]
Since $I(\mu_{\tau} \pm \varepsilon, \hat F_n) \to I_\tau(\mu_{\tau} \pm \varepsilon,F)$ a.s.~as $n \to \infty$, we have a.s.~that $I(\mu_{\tau}-\varepsilon, \hat F_n)>0>I(\mu_{\tau}+\varepsilon, \hat F_n)$ for large $n \in \N$. Since each map
$x \to I(x, \hat F_n), \ n\in\N,$ is continuous and has exactly one zero $\hat \mu_{\tau,n}$, this zero must a.s.~lie between $\mu_{\tau} \pm \varepsilon$ for large $n \in \N$, that is,
\[ \limsup_{n \to \infty} \big|\hat{\mu}_{\tau,n} - \mu_{\tau}(F) \big| \leq \varepsilon \qquad a.s.\quad \forall \ \varepsilon>0,\]
showing (\ref{eq:consistindividual}).   \\
Using Proposition \ref{lem:continuity} (ii) and individual consistency, the classical Glivenco-Cantelli argument may be applied. Let $d = \mu_{\tau_u}(F) - \mu_{\tau_l}(F)$, $m \in \N$, and choose by continuity
$\tau_l = \tau_0 \leq \tau_1 \leq \ldots \leq \tau_m = \tau_u$ such that $\mu_{\tau_k(F)} =  \mu_{\tau_l}(F) + kd/m$, $k=1, \ldots, m$.
By monotonicity, for $\tau_k \leq \tau \leq \tau_{k+1}$,
\[ \hat{\mu}_{\tau,n} - \mu_\tau(F) \leq \hat{\mu}_{\tau_{k+1},n} - \mu_{\tau_{k+1}}(F) + \mu_{\tau_{k+1}}(F) - \mu_{\tau_{k}}(F).\]
Therefore
\[ \sup_{\tau_l \leq \tau \leq \tau_u}\, \big( \hat{\mu}_{\tau,n} - \mu_{\tau}(F)\big)  \leq \max_{0 \leq k \leq m} \big| \hat{\mu}_{\tau_{k},n} - \mu_{\tau_{k}}(F) \big| + d/m.\]
Similarly,
\[ \sup_{\tau_l \leq \tau \leq \tau_u}\, \big(\mu_{\tau}(F) - \hat{\mu}_{\tau,n}\big)  \leq \max_{0 \leq k \leq m} \big| \hat{\mu}_{\tau_{k},n} - \mu_{\tau_{k}}(F) \big| + d/m.\]
Since
\[ \sup_{\tau_l \leq \tau \leq \tau_u}\, \big| \hat{\mu}_{\tau,n} - \mu_{\tau}(F)\big| = \max\Big(\sup_{\tau_l \leq \tau \leq \tau_u}\, \big( \hat{\mu}_{\tau,n} - \mu_{\tau}(F)\big) ,  \sup_{\tau_l \leq \tau \leq \tau_u}\, \big(\mu_{\tau}(F) - \hat{\mu}_{\tau,n}\big) \Big),\]
we have for any $m \in \N$ that
\begin{align*}
\limsup_n \sup_{\tau_l \leq \tau \leq \tau_u}\, \big| \hat{\mu}_{\tau,n} - \mu_{\tau}(F)\big| & \leq \limsup_n \max_{0 \leq k \leq m} \big| \hat{\mu}_{\tau_{k},n} - \mu_{\tau_{k}}(F) \big| + d/m = d/m \quad a.s.
\end{align*}
\end{proof}

We shall derive Theorem \ref{the:asymplinearization} from Theorems 1 and 10 in \citet{arcones}. For convenience, we state a version of these results,  tailored to our needs.

{\bf Theorem} [{\sl Theorems 1 and 10 in \citet{arcones} }] Let $Y,Y_1,Y_2,\ldots$ be i.i.d.~with distribution function $F$.
Let $g:\R^2\to\R$ be a function such that $g(\cdot,\theta): \R\to\R$ is measurable for each $\theta\in\R$.
Let $\hat{\theta}_n$ be a sequence of r.v.'s satisfying
\[
 n^{-1} \sum_{k=1}^n g(Y_k,\hat{\theta}_n) = \inf_{\theta\in\R} n^{-1} \sum_{k=1}^n g(Y_k,\theta).
\]
Suppose that:
\begin{itemize}
\item[(A.1)]
$\hat{\theta}_n \xrightarrow{\prob} \theta_0, \quad \theta_0\in\R.$
\item[(A.2)]
There is a positive constant $V$ such that
\begin{align*}
 E[g(Y,\theta)-g(Y,\theta_0)]= V(\theta-\theta_0)^2 + o(|\theta-\theta_0|^2),
\end{align*}
as $\theta \to\theta_0$.
\item[(A.3)]
Let $\varphi:\R\to\R$ and let $\{a_n\}$ be a sequence of positive numbers which converges to infinity with
$\sup_{n\geq 1} n^{-1} a_n^2 < \infty$ such that
\begin{align*}
a_n \left( n^{-1} \sum_{j=1}^n \varphi(Y_j) - E[\varphi(Y)] \right) =O_{\prob}(1).
\end{align*}
\item[(A.4)]
There is a function $\zeta:\R\to\R$ with $ E|\zeta(Y)|<\infty$ such that
\begin{align*}
 \lim_{\delta \to 0} E\left[ \sup_{|\vartheta|\leq\delta} \frac{|r(Y,\vartheta)-\vartheta^2 \, \zeta(Y)|}{\vartheta^2} \right] = 0,
\end{align*}
where $r(y,\vartheta)= g(y,\vartheta_0+\vartheta) - g(y,\vartheta_0) - \vartheta \varphi(y)$.
\end{itemize}
Then,
\begin{align}
 a_n(\hat{\theta}_n-\theta_0) + \frac{a_n}{2V} \left( \frac{1}{n} \sum_{j=1}^n \varphi(Y_j) - E[\varphi(Y)] \right) &\xrightarrow{\prob} 0.
\end{align}
\begin{proof}[{\sl Proof of Theorem \ref{the:asymplinearization}}]
We verify the conditions of the above theorem for $\theta_0 = \mu_\tau(F)$, $g(y,\theta) = S_\tau(\theta,y)$, $\varphi(y) = - I_\tau\big(\mu_\tau(F),y\big)$ and
\[ \zeta (y) = \frac{\tau}{2}\, 1_{\{y > \mu_\tau\}} + \frac{1-\tau}{2}\, 1_{\{y < \mu_\tau\}} \, . \]
(A1) follows from Theorem \ref{prop:uniformconsistency}. \\
(A2) follows from (\ref{eq:asympcontrast}), (\ref{eq:derivatives}), the assumption of continuity of $F$ at $\mu_\tau(F)$, and Taylor's theorem, which holds under the minimal assumption of an existing second derivative. \\
(A3) is (\ref{eq:boundedness}).\\
Finally, for (A4) we compute that for $x>0$,
%
\begin{eqnarray*}
\lefteqn{ S_\tau(\mu_{\tau}+x,y) - S_\tau(\mu_{\tau},y) + x I_\tau(\mu_{\tau},y) } \\
& = & \ - \frac{\tau}{2}\, (y - \mu_\tau - x)^2\, 1_{\{\mu_\tau < y  \leq \mu_\tau+x\}} + \frac{\tau}{2}\, x^2 1_{\{y > \mu_\tau\}}\\
&& \ + \ \frac{1-\tau}{2}\, (y - \mu_\tau - x)^2\, 1_{\{\mu_\tau \leq y  < \mu_\tau+x\}} + \frac{1-\tau}{2}\, x^2 1_{\{y < \mu_\tau\}}
\end{eqnarray*}
%
and similarly for $x < 0$. Therefore for some $c>0$ we may estimate
	\begin{align*}
\big|S_\tau(\mu_{\tau}+x,y) - S(\mu_{\tau},y) + x I_\tau(\mu_{\tau},y) - x^2 \zeta(y) \big|  %
\leq & c\, (y - \mu_\tau - x)^2\, 1_{\{\mu_\tau - |x| \leq y  \leq  \mu_\tau+|x|\}}\\
\leq &c\, x^2 \, 1_{\{\mu_\tau - |x| \leq y  \leq  \mu_\tau+|x|\}} \, ,
\end{align*}
and therefore
\begin{align*} 
& E\left[ \sup_{|x|\leq\delta}
 \frac{| S_\tau (\mu_{\tau}+x,Y) - S_\tau (\mu_{\tau},Y) + x I_\tau (\mu_{\tau},Y) - x^2 \zeta(Y) |}{x^2} \right]\\
\leq & c\, \prob\big(\mu_\tau - \delta \leq Y \leq \mu_\tau + \delta \big) \to 0,\quad \delta \to 0,
\end{align*}
since $Y$ does not have a point mass at $\mu_\tau$.
\end{proof}
\begin{proof}[{\sl Proof of Theorem \ref{them:asymdirtsecondmom}}]
We start by establishing Lipschitz continuity of $S_\tau(x,y)$ as a function of $x$ with square-integrable Lipschitz constant.
Since
$ \partial_x S_\tau(x,y) =  - I_\tau(x,y),$
we have for $x_1, x_2 \in B_\delta(\mu_\tau)$
\begin{align}\label{eq:Lipschitzprop}
 \left| S_\tau(x_1,y) - S_\tau(x_2,y) \right| \leq c \, m(y) |x_1 - x_2|, \qquad m(y):= \sup_{x \in B_\delta(\mu_\tau)} \big|I_\tau(x,y)\big|.
\end{align}
Then, the inequality

\begin{align*}
 m(y) \leq \sup_{x \in B_\delta(\mu_\tau)} |x-y| \leq \sup_{x \in B_\delta(\mu_\tau)} |x| + |y|
\end{align*}
yields $ E[m(Y)^2] < \infty$ if $EY^2<\infty$, that is, the Lipschitz constant has finite second moment.

Next, the asymptotic contrast in (\ref{eq:asympcontrast}) is continuously differentiable with left and right derivatives in $\mu_\tau$ given in (\ref{eq:derivatives}). From Taylors formula, we obtain
\begin{align}\label{eq:taylortwosided}
\begin{split}
\psi_\tau(x)-\psi_\tau(\mu_{\tau})& =(x-\mu_{\tau})^2 \psi^{''+}_\tau(\mu_{\tau})/2 + o(|x-\mu_{\tau}|^2), \qquad x > \mu_{\tau},\\
\psi_\tau(x)-\psi_\tau(\mu_{\tau})& =(x-\mu_{\tau})^2 \psi^{''-}_\tau(\mu_{\tau})/2 + o(|x-\mu_{\tau}|^2), \qquad x < \mu_{\tau},
\end{split}
\end{align}
where $\psi^{''\pm}_\tau(\mu_{\tau})$ are right/left second derivatives.

Therefore, the assumptions of Theorem 5.52 in \citet{vaart} are satisfied with $\alpha=2$ and $\beta = 1$ (see the argument in Corollary 5.53, that the Lipschitz property (\ref{eq:Lipschitzprop}) implies the concentration inequality), and we obtain the $\sqrt{n}$-rate of convergence:
\[ \sqrt{n} \left(\hat \mu_{\tau,n} -  \mu_{\tau} \right) = O_{\prob} (1).\]
To obtain the asymptotic distribution, we apply the argmax-continuity theorem, Corollary 5.58 in \citet{vaart}. To this end, for a measurable function $f$ with $E f^2(Y)< \infty$, denote
\[ \mathbb{P}_n f = \frac1n \sum_{k=1}^n f(Y_k),\quad P f = E F(Y), \quad \text{and}\quad \mathbb{G}_n(f) = \sqrt{n}\, \big(\mathbb{P}_n  - P \big) f.\]
By the Lipschitz property (\ref{eq:Lipschitzprop}), from the proof of Lemma 19.31 in \citet{vaart} we obtain
for any $M>0$ that
\[ \sup_{|h| \leq M} \mathbb{G}_n\big[\sqrt{n}\,\big(S_\tau(\mu_\tau + h/\sqrt{n},\cdot ) - S_\tau(\mu_\tau,\cdot) \big) + h\, I_\tau(\mu_\tau,\cdot) \big] \stackrel{n \to \infty}{\to} 0 \quad (\prob).\]
Therefore, for any $M>0$, the difference between the processes
\begin{align*}
h \mapsto \sqrt{n}\, \mathbb{P}_n \big[\sqrt{n}\,\big(S_\tau(\mu_\tau + h/\sqrt{n},\cdot ) - S_\tau(\mu_\tau,\cdot) \big) \big], \qquad |h| \leq M,
\end{align*}
and
\begin{align*}
h \mapsto n\, \big[\psi_\tau\big(\mu_\tau + h/\sqrt{n} \big) - \psi_\tau(\mu_\tau) \big) \big] - h\, \mathbb{G}_n I_\tau(\mu_\tau,\cdot), \qquad |h| \leq M,
\end{align*}
tends to $0$ in probability in sup-norm. Using (\ref{eq:taylortwosided}), the second process
converges to the Gaussian process
\begin{align}\label{eq:limitprocess}
h \mapsto \frac{1}{2 \sigma_1}\, h^2 1_{h>0} + \frac{1}{2 \sigma_2}\,\, h^2 1_{h<0}\, - h\, W ,
\end{align}
where $W$ is normally distributed as in the theorem, hence so does the first. From the argmax - continuity theorem, we obtain weak convergence of the minimizers $\sqrt{n} \left(\hat \mu_{\tau,n} -  \mu_{\tau} \right)$ to the minimizer of the limit process.
Now, a parabola $h \mapsto - h W + h^2 / (2 \sigma)$ for some $\sigma>0$ is minimized at $h = \sigma W$, yielding the negative value $- \sigma W^2/2$. Therefore, the minimizer of (\ref{eq:limitprocess}) is at $h = \sigma_1 W$ for $W>0$ and at $h =  \sigma_2 W$ for $W<0$, which gives the statement of the theorem.
\end{proof}
\begin{proof}[{\sl Proof of Theorem \ref{satz:cltexpec}}]
We shall apply \citet[Theorem 1]{vaart95}, which gives asymptotic normality of functional Z-estimators; see also \citet[Theorem 13.4]{kosorok}, which additionally implies validity of the bootstrap. First, Theorem \ref{prop:uniformconsistency} gives the uniform consistency. Given $\nu \in C\big[\tau_l, \tau_u \big] \subset l^\infty\big[\tau_l, \tau_u \big]$, the functions $\tau \mapsto I_\tau\big(\nu(\tau),F \big)$ and $\tau \mapsto I_\tau\big(\nu(\tau),\hat F_n \big)$ are also in $C\big[\tau_l, \tau_u \big]$, and
\[ \tau \mapsto I_\tau\big(\mu_{\tau},F \big) = 0,\qquad \tau \mapsto I_\tau\big(\hat \mu_{\tau,n},\hat F_n \big)=0. \]
Next, we check the conditions (2), (3) and (4) in \citet{vaart95}. Suppose that $\nu \in C\big[\tau_l, \tau_u \big] $ is such that $F$ is continuous on the image of $\nu$ (this is true by our assumption if $\| \nu - \mu\|_{[\tau_l, \tau_u]}$ is small enough). Then we apply the mean value theorem for each $\tau \in [\tau_l, \tau_u]$ to obtain
\begin{align*}
&  \big|I_\tau\big(\nu(\tau),F \big) - I_\tau\big(\mu_\tau,F \big) + \big[\tau \big(1 - F(\mu_\tau) \big) + (1-\tau)\, F(\mu_\tau)\big] \big(\nu(\tau) - \mu_\tau \big)\big|\\
\leq & \big|F(\xi_\tau) - F(\mu_\tau) \big|\, \big|\nu(\tau) - \mu_\tau \big|,
\end{align*}
where $\xi_{\tau}$ is between $\nu(\tau)$ and $\mu_\tau$. Since $F$ is uniformly continuous in a compact neighborhood of $[\tau_l, \tau_u]$, we obtain
\begin{equation}\label{eq:frechet}
\sup_{\tau \in [\tau_l, \tau_u]} \big|I_\tau\big(\nu(\tau),F \big) - I_\tau\big(\mu_\tau,F \big) + \big[\tau \big(1 - F(\mu_\tau) \big) + (1-\tau)\, F(\mu_\tau)\big] \big(\nu(\tau) - \mu_\tau \big)\big|=o\big(\| \nu - \mu\|_{[\tau_l, \tau_u]}\big),
\end{equation}
showing Fr\'echet differentiability, that is, (4) in \citet{vaart95}. Note that the derivative, multiplication with the function $\tau \mapsto - \big[\tau \big(1 - F(\mu_\tau) \big) + (1-\tau)\, F(\mu_\tau)\big]$ is continuously invertible.\\
Since by Proposition \ref{lem:continuity}, (iii), $\mu_\tau(F)$ is continuously differentiable in $\tau$, we have for an appropriate constant $c>0$ that
\begin{equation}\label{eq:lipschitzcondition}
 \big| I_{\tau_1}\big(\mu_{\tau_1}(F), y\big) - I_{\tau_2}\big(\mu_{\tau_2}(F), y\big) \big| \leq c |y|\, |\tau_1 - \tau_2|,\quad \tau_1, \tau_2 \in [\tau_l, \tau_u],\quad y \in \R,
\end{equation}
so that
\begin{equation}\label{eq:functionsDonsker}
 \big(I_\tau\big(\mu_\tau(F), \cdot  \big) \big)_{\tau \in [\tau_l, \tau_u]}
\end{equation}
is a Donsker class of functions, see \citet[Example 19.7.]{vaart}, taking care of (2) in \citet{vaart95}.

Finally, to show (3) in \citet{vaart95}, we choose $\delta_n \downarrow 0$, and estimate in the first step
\begin{align}\label{eq:maximalprocess}
\begin{split}
&\sup_{\| \nu - \mu\|_{[\tau_l, \tau_u]} \leq \delta_n}\, \sup_{\tau \in [\tau_l, \tau_u]}\, \sqrt{n}\, \big|I_\tau\big(\nu(\tau), \hat F_n \big) - I_\tau\big(\nu(\tau),  F \big) - \big[ I_\tau\big(\mu_\tau, \hat F_n \big) - I_\tau\big(\mu_\tau,  F \big) \big] \big|\\
\leq & \sup_{|x| \leq \delta_n}\, \sup_{\tau \in [\tau_l, \tau_u]}\, \sqrt{n}\, \big|I_\tau\big(\mu_\tau + x, \hat F_n \big) - I_\tau\big(\mu_\tau+x,  F \big) - \big[ I_\tau\big(\mu_\tau, \hat F_n \big) - I_\tau\big(\mu_\tau,  F \big) \big] \big| \, .
\end{split}
\end{align}
Now, for a constant $C>0$, $y \in \R$, $\tau, \tau_1, \in [\tau_l, \tau_u]$, $|x|, |x_1| \leq 1$,
\begin{align*}
&\big|I_{\tau_1}\big(\mu_{\tau_1} + x_1, y \big) - I_{\tau_1}\big(\mu_{\tau_1} , y \big) - \big[ I_{\tau}\big(\mu_{\tau} + x, y \big) - I_{\tau}\big(\mu_{\tau} , y \big) \big] \big| \leq C |y| \big(|x_1 - x| + |\tau_1 - \tau| \big),\\
& \big|I_{\tau}\big(\mu_{\tau} + x, y \big) - I_{\tau}\big(\mu_{\tau} , y \big) \leq C \, |y| \, |x|.
\end{align*}
Therefore,
\[ \sup_{|x| \leq \delta_n}\, \sup_{\tau \in [\tau_l, \tau_u]}\, E\big(I_{\tau}\big(\mu_{\tau} + x, Y \big) - I_{\tau}\big(\mu_{\tau} , Y \big) \big)^2 \leq C\, E Y^2\, \delta_n,\]
and each
\[ \cF_n = \Big\{(x,\tau) \mapsto I_{\tau}\big(\mu_{\tau} + x, y \big) - I_{\tau}\big(\mu_{\tau} , y \big),\quad |x| \leq \delta_n,\ \tau \in [\tau_l, \tau_u] \Big\}\]
is a Lipschitz-class of functions. Therefore we may estimate (\ref{eq:maximalprocess}) by the bracketing integral $J_{[]}\big(\delta_n, \cF_n, L_2(F)\big)$ and an additional sequence converging to zero, by \citet[Lemma 19.34 and Example 19.7]{vaart}, which together $\to 0$ as $n \to \infty$ and $\delta_n \downarrow 0$.

Next, we show that weak convergence is actually in $C\big[\tau_l, \tau_u \big]$. The expectile processes (\ref{eq:expectileprocess}) have continuous sample paths.
As for the limiting Gaussian process, it suffices to show continuity of the sample paths of the limit Gaussian process of the empirical process corresponding to the function class (\ref{eq:functionsDonsker}), since the inverse of the Fr\'echet derivative in (\ref{eq:frechet}) is simply multiplication by a fixed continuous function. By \citet[Lemma 18.15]{vaart}, the limit process can be constructed to have continuous sample paths w.r.t.~its standard deviation semimetric. In order to check that continuity also holds w.r.t.~the ordinary distance on $[\tau_l, \tau_u]$, we show that
\begin{equation}\label{eq:continuity}
E\big(I_{\tau_2}(\mu_{\tau_2}, Y  ) - I_{\tau_1}(\mu_{\tau_1}, Y  ) \big)^2 \leq C^2 (\tau_2 - \tau_1)^2,\qquad \tau_j \in \big[\mu_{\tau_l}, \mu_{\tau_u} \big],\ j=1,2
\end{equation}
for some $C >0$. But this follows immediately from (\ref{eq:lipschitzcondition}) upon squaring and integrating. This concludes the proof of the theorem.
\end{proof}
%
%

%
\end{document}